\documentclass{article}

\PassOptionsToPackage{numbers}{natbib}

 \usepackage[preprint]{neurips_2026}

\usepackage[utf8]{inputenc}
\usepackage[T1]{fontenc}
\usepackage{url}
\usepackage{xurl}       %
\usepackage{hyperref}
\usepackage{booktabs}
\usepackage{amsfonts}
\usepackage{amsmath,amssymb}
\usepackage{amsthm}
\theoremstyle{definition}
\newtheorem{proposition}{Proposition}
\usepackage{nicefrac}
\usepackage{microtype}
\usepackage{xcolor}
\usepackage{graphicx}
\usepackage{array}
\usepackage{float}

\title{Do AI Activity Indices Really Measure AI Activity?}

\author{
  Riccardo Bertamini\textsuperscript{1,5,*} \quad
  Hugo M.V. Arsenio\textsuperscript{2,5,*} \quad
  Lorenzo Caputi\textsuperscript{3,5} \quad
  Andrea Procopio\textsuperscript{4,5} \\
  \textsuperscript{1}University of Cambridge \quad
  \textsuperscript{2}Bocconi University \quad
  \textsuperscript{3}Stanford University \\
  \textsuperscript{4}MIT \quad
  \textsuperscript{5}BAINSA \\
  \textsuperscript{*}Equal contribution
}

\begin{document}

\maketitle

\begin{abstract}
  AI dashboards increasingly shape economic debate, yet their headline movements often conceal what changed. For a ratio, the same rise can reflect numerator growth, denominator decline, or both. This ambiguity creates structural non-identification: adding observations or richer fixed effects cannot recover the missing information. We propose three propositions showing that a relative-index regression blends the causal effect of interest with co-movement between focal and aggregate activity and shocks shared with the outcome. These components can give the coefficient either sign under a zero effect, or erase it under a nonzero effect. Publishing denominator levels changes what can be learned by removing the co-movement channel exactly, although shared-shock confounding remains. The problem extends to power: if most identifying variation comes from a handful of units, those units determine what the design can detect, however large the nominal sample looks. At a fixed standard error, the minimum detectable effect diverges as the effective cluster count approaches one, leaving a null result uninformative. We apply this framework to OECD.AI’s hiring index and Eurostat employment across 24-27 European countries, where the detection floor exceeds the reported association. In a refreshed panel, that association disappears once the 2021 survey break is excluded. Ratio indices can still reveal concentration, timing, and co-movement, but establishing displacement requires more; our framework turns that boundary into a practical standard by making denominator levels a minimum disclosure requirement. 
\end{abstract}

\section{Introduction}
Almost everything we know about how AI is spreading through the economy arrives as a ratio: a model family's standing is its share of arena battles; a provider's reach is its share of served requests; adoption is the share of surveyed firms reporting use. The most quoted labour-market indicator, the OECD.AI relative AI hiring index, is the growth of AI-skilled hiring relative to the growth of all hiring \cite{oecdlinkedin}.

The ratio is published so that nuisances such as coverage or the survey frame cancel, and a raw numerator would be no better defined. The problem is not the ratio, but that the provider does not release the two counts. If we ask whether AI is replacing workers, or whether a model family is falling behind, we are given a number that answers neither question, and no subsequent estimator can correct it.
This can be shown with the Chatbot Arena, which samples pairs of models for human comparison \cite{chiang2024} and releases its battle logs, so both counts are observable. Across 106{,}134 battles among 55 models between June and August 2024 \cite{arena100k}, the Claude family's published share of battles fell from 34.6\% to 16.4\%. Over the same window its win rate against the 21 opponents present in both periods moved by $-0.01$ points once the family's own model mix is held fixed, inside a bootstrap interval of $\pm 4$ points. The primary change occurred in the sampling process: fourteen entrants accounted for 57\% of late-window battles. Although daily volume remained stable, the family's own count declined, resulting in a decreased ratio (Figure~\ref{fig:arena}).

These examples motivate a narrow question: what can a regression on a relative activity index establish when the counts inside it are withheld? We give a formal answer, and we measure how the resulting identification and power constraints bind in the regression.

\paragraph{Contributions and scope.} We study regressions of outcomes on published relative indices when the component levels are withheld, and we characterise the additional disclosure required for identification. Instead of estimating the employment effect of AI, we focus on whether the published indices could identify it at any size. Our contributions are:
\begin{itemize}\setlength{\itemsep}{0pt}\setlength{\parskip}{0pt}
\item  We derive an exact decomposition of the coefficient from a regression of any outcome on a published relative index: the causal effect enters summed with a cyclicality term, which is the classical ratio confound, and a shock covariance term. The decomposition survives every fixed-effects design, because each within-transformation adds one equation and one unknown, so the effect is not identified from the index at any magnitude (Proposition~\ref{prop:nonid}).
\item  We prove that publishing the denominator in levels removes the cyclicality term exactly while the covariance term survives, and that where the two have opposite signs disclosure can leave the total bias larger than before. The textbook remedy, entering the ratio alongside its parts \cite{kronmal1993}, is not available to the analyst here, because the parts are proprietary; only the provider can enable it (Proposition~\ref{prop:id}).
\item We establish a lower bound on the smallest detectable effect that depends on the effective number of clusters rather than on the sample size (Proposition~\ref{prop:power}).
\end{itemize}
Empirically, we join the hiring index to Eurostat employment for 24 to 27 European countries from 2018 to 2025. A one-standard-deviation higher index is associated with about 0.5 percentage points lower employment growth, but inference varies across procedures and two or three countries carry most of the signal. In the panel reaching 2025, the association survives only through the 2021 survey break. Its minimum detectable effect is 1.8 to 4.5 times the estimate. The dashboards reveal concentration, timing and co-movement, but not displacement.

\section{Related work}
\label{sec:related}
Pearson showed that indices sharing a denominator correlate even when the underlying quantities do not \cite{pearson1897}, and Kronmal restated the problem for regression: a ratio may enter only alongside its parts as main effects \cite{kronmal1993}. Kuh and Meyer, and then Belsley, developed the same point for deflated variables in econometrics \cite{kuh1955,belsley1972}, and Aitchison's log-ratio analysis is its compositional-data form \cite{aitchison1982}. Our published index is an Aitchison log-ratio, and Proposition~\ref{prop:id} is Kronmal's prescription written out for a setting where the parts are withheld.

The task framework of Autor, Levy and Murnane \cite{autor2003} and the robot-exposure design of Acemoglu and Restrepo \cite{acemoglu2020} set the template that generative AI inherited. Experimental and firm-level work finds productivity gains on cognitive tasks \cite{noy2023,brynjolfsson2025,babina2024}; task-based measures imply broad exposure \cite{eloundou2024,imf2024,nurski2024,felten2021,gmyrek2023,gmyrek2025} against modest aggregate effects \cite{acemoglu2025}; and direct employment evidence is mixed and mostly national \cite{humlum2026,henseke2026,ecbnewtech,georgieff2021,lebastard2026,bis1325,guarascio2025,hampole2025,canaries2025}, focusing on how large the effect is. We ask whether the published indicators can identify that employment effect at any magnitude. Proposition~\ref{prop:nonid} shows they cannot: with only the ratio index and the outcome observed, the estimated coefficient is consistent with every value of the effect: the coefficient can take either sign when the true effect is zero, and be zero when it is not.

Our object is the measurement layer these dashboards constitute \cite{agrawal2019}: OECD.AI's LinkedIn-derived hiring, migration and skills series \cite{oecdlinkedin}, venture-capital aggregates \cite{oecdvc2025}, enterprise adoption surveys \cite{eurostataiadoption} and vacancy rates \cite{estat_jvr}. The same form governs how AI systems themselves are watched: Chatbot Arena publishes shares of pairwise battles \cite{chiang2024,arena100k}, and Singh et al.\ \cite{singh2025} document unequal sampling, selective disclosure of private variants and silent deprecation behind its ratings. It establishes that one platform's numbers are distorted, whereas we show that a published ratio cannot carry a causal reading however well the platform behaves.

When many decision-makers run the same algorithm, its mistakes stop cancelling out and they share the same errors \cite{kleinberg2021}. The same logic applies to measurement: one dashboard supplies most public claims about AI and jobs, so having many readers gives no independent check on it. Separately, a prediction that gets acted on changes what it predicted \cite{perdomo2020}, and people who know they are being scored move to improve the score \cite{hardt2016}. A ratio is especially easy to influence, since it rises either when AI-skilled hiring goes up or when total hiring goes down.

Cluster-robust asymptotics fail when a handful of clusters carry the identifying variation \cite{cgm2008,mackinnon2017,mnw2023}, effective cluster counts diagnose it \cite{carter2017}, and randomization inference trades the reference distribution for an exchangeability assumption \cite{canay2017,young2019}. The shift-share literature supplies the analogous warning for exposure designs \cite{gpss2020,akm2019}, as does pre-testing for parallel trends \cite{roth2022}. 

\section{Methods}
\subsection{Data}
\label{sec:data}

Three providers supply the labour-market analysis, OECD.AI, Eurostat and the ILO, and Table~\ref{tab:data} gives the series and their coverage, including a policy score of our own construction. Section~\ref{sec:model}'s arena exhibit uses a further source, the released Chatbot Arena battle logs \cite{arena100k}, which the table omits. We exclude 2026 everywhere as a partial year.

\begin{table}[H]
  \caption{Primary sources and their coverage. The text below states what each is used for.}
  \label{tab:data}
  \centering
  \footnotesize
  \setlength{\tabcolsep}{4pt}
  \begin{tabular}{@{}>{\raggedright\arraybackslash}p{4.0cm}>{\raggedright\arraybackslash}p{8.6cm}@{}}
    \toprule
    \textbf{Source and variable} & \textbf{Coverage} \\
    \midrule
    OECD.AI relative AI hiring index \cite{ref2.3} & 2018-01 to 2026-01, monthly, averaged to years \\
    OECD.AI AI skills migration \cite{ref2.1} & Per 10{,}000 LinkedIn members; extracts spanning 2019-2024 (Panel A) and 2021-2025 (Panel B) \\
    OECD.AI VC in AI start-ups \cite{ref1.1,ref1.2} & By country and by industry, 2012-2026; the industry table's country dimension is the investor's country \\
    OECD.AI AI skills penetration \cite{ref2.2} & By country and industry, 2016-2024 \\
    Eurostat employment by occupation \cite{ref3,estat_egai2d} & Ages 15-64; Panel A uses an extract of table lfsa\_egai2d, the other samples lfsa\_epgais over 2016-2025 \\
    ILO generative-AI exposure scores \cite{gmyrek2025} & 427 ISCO-08 unit groups, aggregated to major groups with Eurostat employment weights \\
    Eurostat employment by NACE two-digit \cite{estat_nace} & Ages 20-64, 2018-2025 \\
    Eurostat covariates \cite{estat_edat,estat_gdp,estat_une,estat_ict,estat_stem} and own policy score & Country-year; policy score hand-coded 0-3 (one point each: national AI strategy in force, major update, EU AI Act in force) \\
    \bottomrule
  \end{tabular}
\end{table}

No single join spans 2019 to 2025, so we build two samples. Panel A (24 countries, 2019 to 2024, 144 country-years) has the longer window and the full covariate set and is the primary test; Panel B (26 countries, 2021 to 2025, 130 country-years) is the only sample that observes 2025 and the only one industry investment joins to. They overlap in 88 of Panel B's 130 country-years but differ in window, covariate set and index extract, so Panel B tests Panel A's coefficient on a different extract rather than resampling it, and a pooled panel would have concealed the disagreement. The occupation designs use no migration series, which is why they reach 27 countries.

\subsection{Identification}
\label{sec:model}
Dashboards report activity as a ratio. Let $A_{i,t}$ be a focal count and $B_{i,t}$ an aggregate that contains it, both observed on one platform, and let the published series be $R_{i,t} = \Delta \log A_{i,t} - \Delta \log B_{i,t}$. For the hiring index $A$ is AI-skilled hires and $B$ total hires; for a leaderboard $A$ is one family's battles and $B$ all battles. This object has two key properties.
\begin{itemize}\setlength{\itemsep}{0pt}\setlength{\parskip}{0pt}
\item  Invariant to whatever scales $A$ and $B$ together, which is exactly why providers publish it: platform coverage, the member base and any common reporting convention cancel. 
\item Level information is discarded by the invariance: the same value may reflect a collapsing numerator, a growing denominator, or both, and the series cannot distinguish among them.

\end{itemize}

Two departures from this idealised object matter.
\begin{itemize}\setlength{\itemsep}{0pt}\setlength{\parskip}{0pt}
\item The published index is smoothed, so $R = \Delta\log A - \Delta\log B + \omega$,
where the slack $\omega$ is one further free nuisance: it leaves
Proposition~\ref{prop:nonid} untouched but weakens Proposition~\ref{prop:id}, since
disclosing $B$ then recovers the AI-specific shock only up to $\omega$.
\item This instrument is published as a growth differential, whilst most others of this kind are published as level shares $S = A/B$. The two are one object, since
$\Delta \log S = \Delta \log A - \Delta \log B$, so a regression on $\Delta \log S$
inherits \eqref{eq:decomp} term for term, and one on $\Delta S$ differs by a unit-
and period-specific rescaling that changes no term's status.
\end{itemize}
Consequently, the results apply whether the provider publishes a level share
or a growth differential.

For the hiring index we can argue that the two counts move for unrelated reasons but cannot show it, because neither is published. By contrast, the same analysis is possible for Chatbot Arena because it releases its battle logs \cite{chiang2024}. The figures reported in the introduction locate the change in the numerator. Daily battle volume remained flat, and the roster fell from 41 models to 38, ruling out an expansion of $B$. Yet observing $B$ would establish only that the movement occurred through $A$, not why $A$ moved, as Proposition~\ref{prop:id} shows. And holding opponents fixed, but not the family's own model mix, yields $+2.5$ points (95\% CI $-0.4$ to $+5.5$) rather than $-0.01$, showing that the performance measure is sensitive to changes in within-family composition (Figure~\ref{fig:arena}).

\begin{figure}[t]
  \centering
  \includegraphics[width=0.93\linewidth]{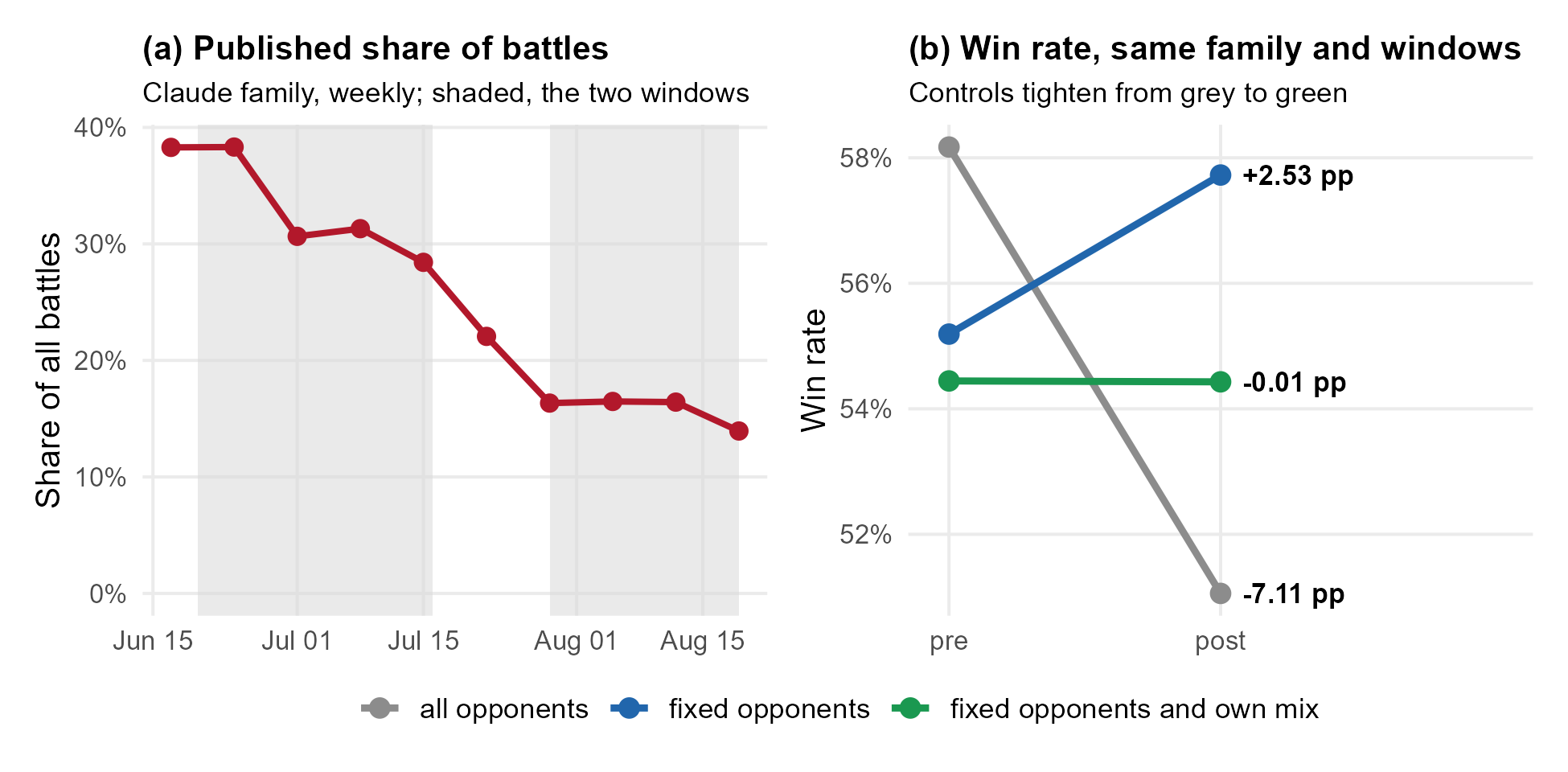}
  \caption{A published index and the signal it is read as. The Claude family's share of Chatbot Arena battles halves as entrants take 57\% of a flat daily volume, while its win rate is flat once the opponent set and the family's own model mix are both held fixed \cite{arena100k,repo}.}
  \label{fig:arena}
\end{figure}

The leaderboard's headline is a Bradley-Terry rating, and it does not escape the problem, since it conditions on a pairing policy and prompt distribution the platform does not publish \cite{singh2025}. The required disclosure is the analogous one: per-family battle counts, the eligible model set per period, and the sampling policy, in levels.

We develop the algebra for the hiring case because the employment data needed to evaluate displacement are available and the index's denominator enters directly into employment dynamics. The identification problem therefore arises from the index itself rather than from a lack of outcome data. The identity $E_{i,t} = E_{i,t-1} + H_{i,t} - S_{i,t}$ makes aggregate hiring a common component of both the index and employment growth $g_{i,t}$, while the numerator measures the AI-skilled hiring margin whose effect on employment is the parameter of interest. Assume finite second moments and a common relative-cyclicality coefficient $\gamma$ across the units entering the pooled moments below. Write the within-country projection of AI-skilled on total hiring as $\Delta \log h = \gamma\, \Delta \log H + u$, where $\gamma$ is a relative cyclicality parameter and $u$ an AI-specific shock orthogonal to $\Delta \log H$ by construction, and let
\begin{equation}
\label{eq:outcome}
g_{i,t} \;=\; \delta\, u_{i,t} \;+\; \eta_{i,t},
\end{equation}
where $\delta$ is the displacement parameter of interest and $\eta$ collects everything else, including the cycle. We define $\delta$ outside \eqref{eq:outcome}, as the causal effect of intervening on the AI-specific hiring shock while holding $\eta$, and with it aggregate hiring, fixed; \eqref{eq:outcome} is a linearity restriction on that structure, not a definition of $\eta$, which would make it an accounting identity holding for any $\delta$. We place no restriction on $\mathrm{Cov}(u, \eta)$: a common shock can move AI-specific hiring and employment together without either causing the other. Substituting the projection into the index gives $R = (\gamma-1)\Delta \log H + u$, and hence
\begin{equation}
\label{eq:decomp}
\underbrace{\frac{\mathrm{Cov}(R, g)}{\mathrm{Var}(R)}}_{\textstyle \beta,\ \text{estimated}} \;=\; \underbrace{\frac{(\gamma-1)\,\mathrm{Cov}(\Delta \log H, g)}{\mathrm{Var}(R)}}_{\textstyle \text{relative cyclicality}} \;+\; \underbrace{\frac{\delta\,\mathrm{Var}(u)}{\mathrm{Var}(R)}}_{\textstyle \text{displacement}} \;+\; \underbrace{\frac{\mathrm{Cov}(u, \eta)}{\mathrm{Var}(R)}}_{\textstyle \text{common shock}}.
\end{equation}
Our only object of interest is the displacement term.

\begin{proposition}[Non-identification]
\label{prop:nonid}
Let $\beta = \mathrm{Cov}(R,g)/\mathrm{Var}(R)$, and let the data identify only $\beta$ and $\mathrm{Var}(R)$. Then $\delta$ is not identified from $(R,g)$:
\begin{equation*}
\forall\, \delta \in \mathbb{R},\ \forall\, b \in \mathbb{R},\quad \exists\, \gamma \in \mathbb{R},\ \exists\, c = \mathrm{Cov}(u,\eta) \in \mathbb{R} : \ \eqref{eq:decomp} \text{ holds with } \beta = b.
\end{equation*}
In particular $\beta$ takes either sign with $\delta = 0$, and $\beta = 0$ is attainable with $\delta \neq 0$. The statement is unchanged under any linear within-transformation $M$, including country and year fixed effects and country-specific trends.
\end{proposition}

\begin{proof}
Write $C_{Hg} = \mathrm{Cov}(\Delta\log H, g)$ and $V = \mathrm{Var}(R)$, so \eqref{eq:decomp} reads $\beta V = (\gamma-1)C_{Hg} + \delta\,\mathrm{Var}(u) + c$: one equation, three unknowns $(\gamma, \delta, c)$. Given any target $(\delta, b)$, set $\gamma = 1$ and $c = bV - \delta\,\mathrm{Var}(u)$; then $\beta = b$. Taking $\delta = 0$, $c = bV$ gives $\beta = b$ of either sign; taking $b = 0$, $c = -\delta\,\mathrm{Var}(u) \neq 0$ gives $\beta = 0$ with $\delta \neq 0$. Equivalently, any orthogonal split of $R$ is admissible, and $\delta$ rescales freely as $\mathrm{Var}(u) \to 0$. For the fixed-effects claim let $M$ be symmetric and idempotent, removing country and year means or any further within-country regressors. Then $MR = (\gamma-1)M\Delta\log H + Mu$ and $Mg = \delta Mu + M\eta$, so
\begin{equation*}
\beta_M \;=\; \frac{(\gamma-1)\,\mathrm{Cov}(M\Delta\log H,\, Mg) \;+\; \delta\,\mathrm{Var}(Mu) \;+\; \mathrm{Cov}(Mu,\, M\eta)}{\mathrm{Var}(MR)},
\end{equation*}
which is \eqref{eq:decomp} with transformed moments and a \emph{new} free nuisance $\mathrm{Cov}(Mu, M\eta)$. Each $M$ buys one equation and one unknown, and adds no restriction unless $\mathrm{Cov}(u,\eta)$ is itself restricted.
\end{proof}

Three restrictions bear on the proposition; the first two jointly overturn it, the third alone suffices:
\begin{itemize}\setlength{\itemsep}{0pt}\setlength{\parskip}{0pt}
\item $\gamma = 1$, so AI-skilled and total hiring share a cycle exactly, eliminates the first term, but it is an assumption about the very quantity the index conceals and no published series tests it. 
\item $\mathrm{Cov}(u,\eta)=0$, asserting that nothing moving employment also moves AI-specific hiring; however, European hiring cooled and AI-specific hiring fell together over 2022 and 2023.
\item Using an instrument that shifts $u$ while moving neither $\eta$ nor the withheld denominator $\Delta\log H$. We find none in the published data; the likeliest candidates, venture capital and skills migration, may affect employment through channels other than the AI-specific hiring shock, and both are procyclical, so they cannot be assumed to shift $u$ alone.
\end{itemize}

Section~\ref{sec:related} places this decomposition in the ratio literature, leaving the sign of the cyclicality term open. The sign of $\gamma$ is not determined a priori: adjustment costs make specialist hiring lumpy, arguing for $\gamma < 1$ and a negative $\beta$ with no displacement at all, while specialist hiring fell further than aggregate hiring in 2022 and 2023, arguing for $\gamma > 1$. Published data cannot separate the two. 
Figure~\ref{fig:ident} shows how, with $\delta$ held at zero, the coefficient changes sign once
$\mathrm{Cov}(u,\eta)$ crosses $-(\gamma-1)\mathrm{Cov}(\Delta\log H,\eta)$. Panel a is not calibrated to the data, so it establishes only that the bias can run
in either direction.

\begin{figure}[t]
  \centering
  \includegraphics[width=0.82\linewidth]{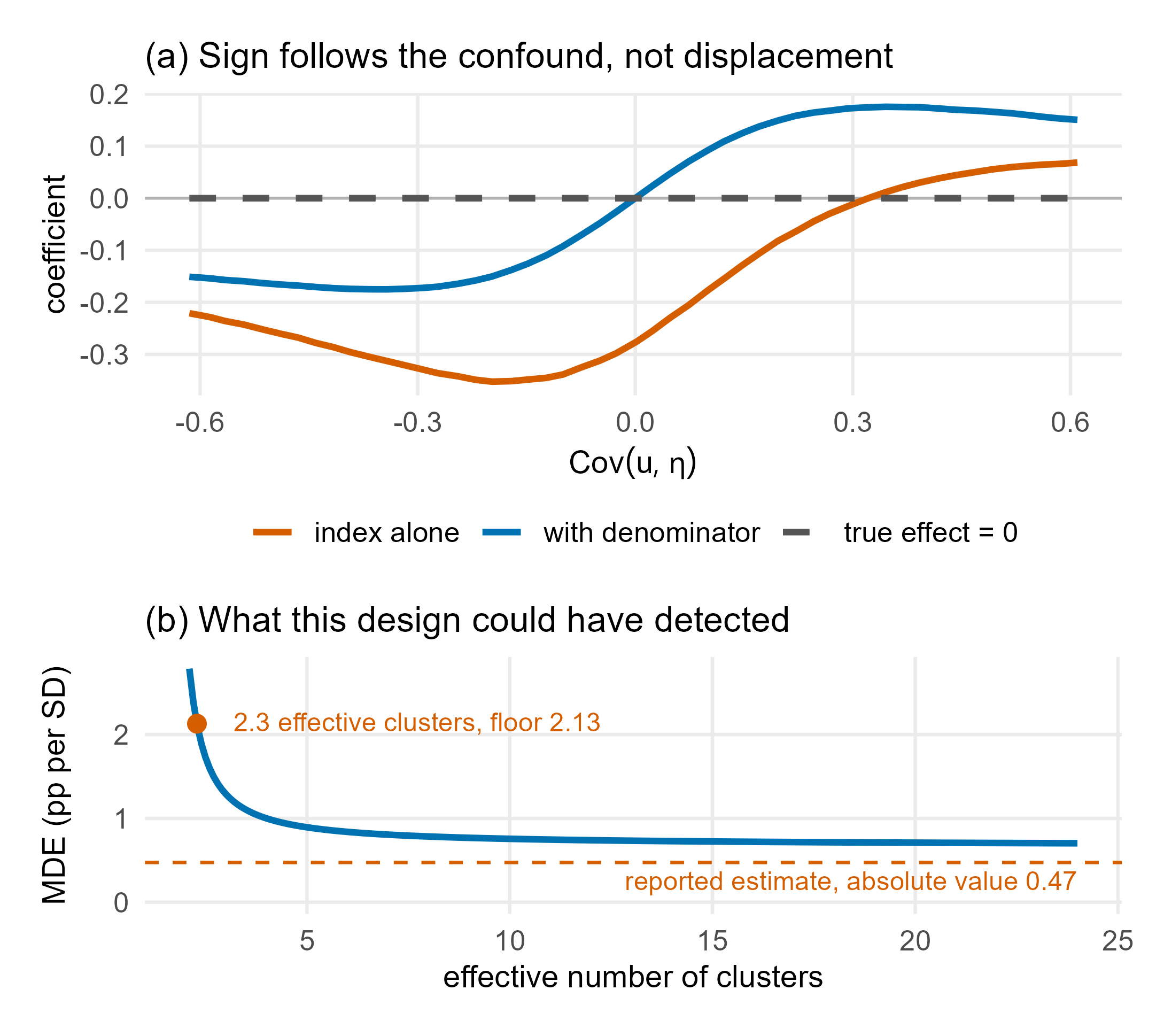}
  \caption{(a) With $\delta=0$ everywhere, the published-index coefficient changes sign as $\mathrm{Cov}(u,\eta)$ varies; disclosure removes the cyclicality channel only, and over part of the range sits further from the truth (simulation details in Appendix~\ref{app:aux}). (b) The Proposition~\ref{prop:power} detection floor at $\sigma=0.2405$: the minimum detectable effect at 80\% power against the effective cluster count, with this design marked at $G^{*}=2.3$ (floor 2.13) and the reported estimate 0.47 dashed.}
  \label{fig:ident}
\end{figure}

\begin{proposition}[What disclosure buys, and what it does not]
\label{prop:id}
Suppose the provider also publishes $\Delta \log B$, here total platform hiring. Then $\gamma$ and $u$ are identified, and the coefficient on $u$ in the projection of $g$ on $(\Delta \log H, u)$ equals $\delta + \mathrm{Cov}(u,\eta)/\mathrm{Var}(u)$. The relative-cyclicality channel is removed exactly and $\gamma$ becomes measurable; $\delta$ remains confounded by $\mathrm{Cov}(u,\eta)$ and is still not identified.
\end{proposition}

\begin{proof}
Given $R$ and $\Delta \log H$, the series $\Delta \log h = R + \Delta \log H$ is observed, so $\gamma$ is the projection coefficient of one observed series on another and $u$ is its residual. Regressing $g$ on $(\Delta \log H, u)$ with $u \perp \Delta \log H$ gives a coefficient on $u$ of $\mathrm{Cov}(u,g)/\mathrm{Var}(u)$, which by \eqref{eq:outcome} equals $\delta + \mathrm{Cov}(u,\eta)/\mathrm{Var}(u)$.
\end{proof}
The removal is exact because u becomes the residual of an observed projection. Disclosure does not however shrink the bias: where the channels have opposite signs, removing one enlarges the net error, as Figure~\ref{fig:ident} shows, so a disclosed index can sit further from the truth than the ratio it replaced while being far more informative about why. One caveat is substantive: the provider publishes rates, not counts, so recovering $\Delta\log H$ also needs the member-base growth that cancels out in the ratio. The required disclosure is therefore the underlying levels. Even with those levels, however, identification is not the only constraint.

\begin{proposition}[Detection floor]
\label{prop:power}
For a clustered design with standard error $\sigma$ and $G^{*}$ effective clusters, the smallest effect detectable at the 5\% level with power $1-\kappa$ is approximately $\mathrm{MDE} = \left(t_{0.975,\,G^{*}-1} + t_{1-\kappa,\,G^{*}-1}\right)\sigma$, which diverges as $G^{*} \to 1$ at fixed $\sigma$, irrespective of the number of observations.
\end{proposition}

$G^{*}$ is the effective cluster count implied by the concentration of the regression's own influence: it falls below the nominal count whenever a few clusters carry most of the squared score mass, so it is a property of the design: adding observations to the same few clusters does not raise it; only
variation spread across more clusters would, and the panel already spans the continent.

\section{Results}
\subsection{Dashboards}
\label{sec:pipeline}
Financing is deal-concentrated, though less than the aggregate share implies. IT infrastructure takes 65.21\% of European AI venture capital in 2025, but one investor-country cell carries most of that, and the share falls to 27.8\% on EU27 member rows, where there is no rebound  above the 2021 level. One robust finding remains: mega-deals were about 73\% of global AI investment value \cite{oecdvc2025}. Appendix~\ref{app:desc} gives the series and the EU27 aggregation gap.

The published index increases in 2024, eighteen months after the ChatGPT release~\cite{chatgptlaunch}.   %
We hold two extracts of the index: the medians here use the refreshed one and Panel A's regressions the other, so agreement between them is corroboration only within one data ecosystem. The median across 27 countries ran 3.30, 1.48, 3.43 over 2018 to 2020, turned negative at $-1.31$, $-4.93$, $-4.32$ over 2021 to 2023, then rose
to 6.80 and 7.96. The upturn is common to both extracts and survives changes of country set, means for medians, and no year-averaging. The pre-2021 level does not, so the rise partly recovers the earlier level rather than establishing a new one.

Capital and talent co-move, and employment reallocates. Lagged venture capital goes with a higher hiring index the following year under country and year effects (3.60, SE 1.73, $p = 0.048$), but not without them and not on the industry panel, and the outcome in that model is the index itself, so Proposition~\ref{prop:nonid} bounds it too (Figure~\ref{fig:industry}). Employment shares shift toward IT infrastructure and healthcare and away from the residual. These are shares, so they describe reallocation, and the programming trend predates generative AI \cite{eurostatict}. Appendix~\ref{app:desc} gives the coefficients and the shares.

Treating November 2022 as the event date, three of four annual designs reject their pre-trend test; the fourth, whose outcome is employment growth rather than the index, does not ($p = 0.100$) and shows nothing after the date. A shift-share sensitivity agrees, since its bivariate first stage flips sign once controls enter \cite{gpss2020,akm2019,roth2022}. All four designs reported in Appendix~\ref{app:aux}.

\subsection{AI hiring and employment growth}
\label{sec:country}
Proposition~\ref{prop:nonid} says a coefficient of either sign is compatible with no displacement. We nonetheless estimate the regression and assess what it can support. The outcome $Y_{i,t}$ is total employment growth in percent year on year, and the baseline is
  \begin{equation*}
  Y_{i,t}
  = \alpha_i + \lambda_t
  + \beta_1 \text{Hiring}_{i,t}
  + \beta_2 \text{Migration}_{i,t}
  + \beta_3 \text{Tertiary}_{i,t}
  + \beta_4 \text{Policy}_{i,t}
  + \varepsilon_{i,t},
  \end{equation*}
  with errors clustered by country. Twenty-four clusters is too few for reliable
  asymptotics, so we fix four procedures in advance: CRV1 intervals
  ($t(23)$)\footnote{The package supports two conventions. CRV1 charges absorbed fixed
  effects to the degrees of freedom, giving $\mathrm{SE}=0.241$ and $p=0.061$, whereas
  the package default gives $\mathrm{SE}=0.219$ and $p=0.041$. We use CRV1 for every
  reported result.}; a wild cluster bootstrap-$t$ with Rademacher weights, null imposed,
  $B=99{,}999$, and seed 42 \cite{cgm2008,mackinnon2017}; the same $t$ referred to the
  effective cluster count; and studentized randomization inference
  \cite{canay2017,young2019}. A 200-variant specification curve leaves 198 estimates
  negative, but 192 of the 200 perturb one 144-observation regression (Appendix~\ref{app:spec}).

\begin{table}[H]
\caption{Robustness of the relative-AI-hiring coefficient (per 1 SD), Panel A. Intervals are CRV1 $t(23)$; $p$ columns are analytic clustered and wild cluster bootstrap-$t$ (99{,}999 draws for the baseline, 4{,}999 elsewhere, one shared draw sequence). Remaining preprocessing rows and full influence and inference diagnostics are in Appendix~\ref{app:aux}. The text reports two further reference distributions \cite{repo}.}
  \label{tab:robust}
  \centering
  \footnotesize
  \setlength{\tabcolsep}{4pt}
  \begin{tabular}{@{}>{\raggedright\arraybackslash}p{4.6cm}cccc@{}}
    \toprule
    Check & Coefficient & 95\% interval & $p_{\text{analytic}}$ & $p_{\text{bootstrap}}$ \\
    \midrule
    Two-way FE, baseline & $-0.47$ & $[-0.97, +0.02]$ & $0.061$ & $0.037$ \\
    $+$ GDP-growth and unemployment controls & $-0.46$ & $[-0.83, -0.09]$ & $0.017$ & $0.014$ \\
    Without COVID years ($n = 96$) & $-0.45$ & $[-0.99, +0.09]$ & $0.097$ & $0.003$ \\
    Without 2021, the LFS break year ($n = 120$) & $-0.45$ & $[-1.07, +0.18]$ & $0.157$ & $0.030$ \\
    \bottomrule
  \end{tabular}
\end{table}

Relative AI hiring is the only baseline regressor with a recurring partial association: a 1 SD higher index goes with roughly 0.3 to 0.6 percentage points lower employment growth within a country, while migration, tertiary share and policy score show none. Table~\ref{tab:robust} reports four specifications; the coefficient stays negative throughout, but the 95\% intervals reach or cross zero in most rows. Employment weighting strengthens it to $-0.89$ per SD ($p = 0.004$) and country trends shrink it to $-0.34$ ($p = 0.17$). Adding GDP growth and unemployment leaves it unchanged: the channel Proposition~\ref{prop:nonid} formalises runs through total hiring, which neither control measures.

The identifying variation is concentrated in very few clusters based on the decomposition: T\"urkiye carries 47\% of the squared mass on the restricted scores the bootstrap resamples, giving 3.2 effective clusters, while on the unrestricted scores Estonia carries 61\% and $G^{*} = 2.3$. Estonia and Latvia hold the panel's two largest regressor values, $+4.6$ and $+4.4$ standard deviations in 2019, so small-country volatility affects hiring as well as migration. Leave-one-cluster-out agrees: dropping T\"urkiye moves the baseline from $-0.47$ ($p = 0.061$) to $-0.35$, Latvia to $-0.50$ and Estonia to $-0.73$ ($p = 0.005$), so single-cluster deletions span $-0.35$ to $-0.73$.

As the estimate depends on the cluster, the p-value depends on the procedure. Each $p_{\text{bootstrap}}$ value in Table~\ref{tab:robust} undercuts its analytic counterpart $p_{\text{analytic}}$, in one row by 0.127, while referring the same $t$ to $t(G^{*}-1)$ gives $p = 0.25$ unrestricted and $p = 0.17$ restricted. It might seem that the bootstrap over-rejects and $t(G^{*}-1)$ corrects it; however, simulated under the null on this design's own matrix, both the bootstrap and CRV1 are correctly sized, so their gap in Table~\ref{tab:robust} is unexplained, while $t(G^{*}-1)$ rejects in none of the 353 replications as concentrated as the data, so $p = 0.25$ is uninformative, which restates Proposition~\ref{prop:power}. Appendix~\ref{app:aux} gives the calibration in full, with two cautions about $G^{*}$: it is estimator-dependent (5.5 under an alternative weighting, $p = 0.11$) and noisy (median 3.4, 5 to 95\% range 1.8 to 7.2). Studentized randomization inference \cite{canay2017,young2019} gives $p = 0.045$ over 20,000 draws and is the least sensitive of the four to the error process, though the exchangeability it rests on strains this panel, since the heterogeneity driving the leave-one-out range is what a permutation null denies. Across procedures, the baseline runs from 0.037 to 0.25, but the 0.25 end is uninformative: it comes from a reference distribution on which this design could never have rejected. We do not need the estimate to be insignificant: Proposition~\ref{prop:nonid} bars the displacement reading at every value.

At the reported $\sigma = 0.2405$ and $G^{*}$ between 2.29 and 5.5, the Proposition~\ref{prop:power} floor at 80\% power runs from 2.13 down to 0.86, and to 0.70 with all 24 clusters effective (Figure~\ref{fig:ident}b); every value exceeds the 0.47 the primary specification reports. Inverting the noncentral $t$ gives exactly 2.38 and 0.87, so the closed form understates the
constraint at the concentrated end. The comparison is not fixed in advance, since $\sigma$ and $G^{*}$ are both estimated from the regression being evaluated, so a floor above the estimate is implied by $p > 0.05$. That ratio, 1.8 to 4.5 times the estimate, therefore rescales the null result rather than testing it. Against a benchmark fixed in advance, even the floor's best case, 0.70 points of annual employment growth per standard deviation, exceeds any employment effect documented for earlier automation waves \cite{acemoglu2020}, so effects of every previously observed size were undetectable by design.

The natural probe of Proposition~\ref{prop:nonid} is a measure whose denominator is not total hiring \cite{eurostataiadoption,ref2.4,estat_jvr}. AI talent concentration, whose denominator is platform membership, is published in three encodings on the same 221 observations. Only one has power against the baseline effect (90\%), and it is null: $+0.02$ per SD, $p = 0.91$. The other two are underpowered at 70\% and 13\%,  so their negative point estimates are not evidence (Appendix~\ref{app:aux}), and two further probes cannot discriminate at this sample size either (Appendix~\ref{app:aux}). The second check is Panel B: with predictors in their own units and venture capital as $\log(1+\text{deals})$, the hiring association holds in sign but weakens to $-0.027$ under two-way FE ($p = 0.059$, bootstrap $0.056$), about $-0.24$ per SD in Panel A. Every 2021 observation carries Eurostat's break flag from the Labour Force Survey redesign, and 2021 growth is measured against the 2020 depression; excluding 2021 moves the coefficient to $+0.002$ ($t = 0.07$) and excluding all break-flags to $+0.003$, while the same exclusion leaves Panel A at $-0.45$ (Figure~\ref{fig:replication}, Appendix~\ref{app:sources}).

The two extracts of the index are not the same series. On the $208$ country-years where they overlap they correlate $0.447$; one is positive in $200$ of those cells and the other in $119$, so they disagree in sign on at least $81$. This undocumented divergence is itself evidence: under the two-way demeaning the baseline uses, the extracts correlate $0.528$, so the identifying variation behind the $-0.47$ is only half-reproducible.

 Local projections \cite{jorda2005} give cumulative associations between $-0.36$ and $-0.62$ over horizons 0 to 3, with only $h = 2$ excluding zero under CRV1, and the lead estimate is imprecise but larger than the contemporaneous one, so reverse causality is not excluded (Appendix~\ref{app:aux}).

\subsection{Occupations}
\label{sec:occ}
A displacement hypothesis predicts losses concentrated in occupations more exposed to AI, which the country regressions cannot see. We test it on a country $\times$ occupation $\times$ year panel of nine ISCO-08 major groups, regressing occupation employment growth on the hiring index interacted with exposure, under country $\times$ year, country $\times$ occupation and occupation $\times$ year fixed effects, clustered by country. Exposure is the index of the published ILO Working Paper \cite{gmyrek2025}, 427 unit groups aggregated to nine with Eurostat employment weights; an unweighted crosswalk and three author-coded rank proxies \cite{felten2021,gmyrek2023} are robustness encodings.

The country $\times$ year effects absorb every country-level shock in levels, so identification is cross-occupation within country-year, but they do \emph{not} absorb the channel of Section~\ref{sec:model}: substituting the decomposition gives $R_{i,t}X_o = (\gamma-1)\Delta \log H_{i,t}X_o + u_{i,t}X_o$, and a country-year term times an occupation characteristic survives those fixed effects exactly as the displacement term does. Exposure indices load on cognitive and clerical content, where cyclical sensitivity also differs.

The interaction is negative across five exposure encodings and an employment-weighted variant, ranging from $-0.34$ to $-1.69$. Two pass the 5\% threshold under CRV1 and one under the wild bootstrap. Excluding the COVID years, the largest economy, or 2021 preserves the sign. With $G^{*}=4.7$, referring the statistic to $t(G^{*}-1)$ gives $p=0.16$. Appendix~\ref{app:desc} reports all estimates and $p$-values.

 Fourteen of the 1{,}458 observations show annual growth above 20 log points, mostly 2021 survey-redesign values; trimming them halves the published-index estimate. The occupational levels, however, run in the opposite direction: growth rises monotonically with exposure in both eras, with the era-over-era gain largest for the \emph{least} exposed tercile (Figure~\ref{fig:terciles}), so the negative interaction is a within-country-year co-movement with the index.

\section{Limitations}
Relative AI hiring measures how fast AI-skilled hiring changes relative to overall hiring, not the number of AI jobs. Skills migration is per 10{,}000 LinkedIn members, so small-country rates are volatile. The venture-capital data measure financing activity, so large deals can dominate.

Proposition~\ref{prop:nonid} formalises one mechanism under a linear projection. It does not model platform selection or measurement error, and the published index is a smoothed percentage change, so Proposition~\ref{prop:id} describes an idealised series. Our $\delta$ is a direct effect holding $\eta$, and so aggregate hiring, fixed: displacement operating through total hiring loads onto $\gamma$, so Proposition~\ref{prop:id} is not a route to a total effect, which is not identified either.

 Net employment cannot reveal gross displacement \cite{davis1992}: separations, hours, wages and task reallocation may move while headcount does not. Every AI-activity measure here derives from LinkedIn's member base, and the arena exhibit from a second platform with its own sampler, so platform composition can move the indicators independently of the quantity of interest. Our nulls are failures to reject.

\section{Discussion}
This work asked what a published activity ratio can identify. A regression of any outcome on the index returns the sum of three terms, only one of which is the causal effect, and no fixed-effects design separates them (Proposition~\ref{prop:nonid}). Publishing the denominator in levels removes the cyclicality term exactly and leaves the shared shock in place (Proposition~\ref{prop:id}), and the effective number of clusters bounds what any such design could detect (Proposition~\ref{prop:power}). Both exhibits show these limits binding: the Arena share halved with the win rate flat because the sampling changed, and the employment association rests on two or three countries, sits below its own detection floor, and reproduces on the 2025 panel only through the 2021 survey break. The dashboards do support concentration, timing and co-movement, as Section~\ref{sec:pipeline} reports.

More broadly, the decomposition applies wherever an ecosystem is watched through ratios, and each result names the disclosure that would answer it. Proposition~\ref{prop:nonid} sets the disclosure minimum: both counts in levels, at the frequency of the ratio, including the base whose growth cancels in the published rate. The Arena exhibit adds the sampling layer: which units were eligible in each period and how they were sampled, without which a share confounds performance with exposure. Section~\ref{sec:country} adds versioning: the divergence we document arrived
with no published notice, so extracts should be citable by version.
Proposition~\ref{prop:power} then obliges the analyst rather than the provider: report the effective number of clusters and the smallest detectable effect alongside any estimate built on such an index.

\appendix

\section{Dashboard descriptives and occupation results}
\label{app:desc}

This appendix reports the arithmetic behind the descriptive statements and the
full occupation results.  These estimates document concentration,
co-movement, and reallocation; none identifies the numerator-specific
employment effect defined in Section~\ref{sec:model}.

\paragraph{Venture-capital concentration.}
Table~\ref{tab:app_vc} shows why the headline European aggregate is sensitive
to its geographic base.  The country field in this extract denotes the
investor side of the flow.  The United Kingdom by IT-infrastructure cell alone
accounts for approximately 54\% of the 2025 all-country total.  Moreover, the
provider's EU27 aggregate is approximately 2.8 times the sum of the member-state
rows.  We cannot reconcile those two series, so we report both rather than
silently selecting one.

\begin{table}[t]
  \centering
  \small
  \setlength{\tabcolsep}{5pt}
  \caption{European AI venture capital by geographic base (USD millions).}
  \label{tab:app_vc}
  \begin{tabular}{lrrrr}
    \toprule
    Geographic base & 2021 total & 2024 total & 2025 total & IT share, 2025 \\
    \midrule
    All country rows       & 41{,}074 & 23{,}654 & 62{,}015 & 65.2\% \\
    Excluding United Kingdom & 18{,}225 & 13{,}170 & 18{,}740 & 37.1\% \\
    EU27 member-state rows & 14{,}837 & 10{,}611 & 13{,}835 & 27.8\% \\
    Provider's EU27 aggregate row & 42{,}039 & 29{,}996 & 39{,}505 & 27.2\% \\
    \bottomrule
  \end{tabular}
\end{table}

\paragraph{Capital, hiring, and employment.}
In the 2018-2025 descriptive panel, lagged log AI venture capital is positively
associated with the following year's relative-AI-hiring index when country and
year effects are included (coefficient 3.60, SE 1.73, $p=0.048$).  The
association is absent without those effects ($-0.22$, $p=0.48$), and lagged
own-industry venture capital does not predict industry employment growth
($-0.388$, $p=0.299$).  These specifications therefore establish
co-movement, not a capital-to-employment mechanism.

On the 33 countries observed at both endpoints, employment shares between 2018
and 2025 change by $+0.960$ percentage points for IT infrastructure and
hosting, $+0.904$ for health care, drugs, and biotechnology, and $-1.479$ for
the residual category.  Pooling a changing country set would instead give
$+0.845$, $+0.450$, and $-0.774$ points, respectively.  We use the balanced
figures because the United Kingdom is 11.7\% of pooled 2018 employment and is
absent in 2025.

On the industry panel, penalised and boosted learners evaluated by
leave-one-country-out cross-validation do not improve on least squares (RMSE
0.226 for the lasso and 0.239 for gradient boosting against 0.224 for OLS and
0.278 for the mean baseline); the lasso penalty is tuned within each training
fold and the boosting rounds by country-grouped inner cross-validation.

\paragraph{Occupation interaction.}
The estimating equation is
\[
  \Delta\log E_{iot}
    = \theta\,\bigl(R_{it}X_o\bigr)
      + \alpha_{it}+\mu_{io}+\lambda_{ot}+\varepsilon_{iot},
\]
where $R_{it}$ is standardised relative AI hiring and $X_o$ is occupation exposure.  The fixed effects are country-year, country-occupation, and occupation-year; standard errors are clustered by country.  Continuous exposure measures are standardised. The cognitive/manual contrast is binary.
Table~\ref{tab:app_occ} distinguishes five exposure encodings from the third row, which is a lagged-employment-weighted regression using the same published ILO exposure measure.

\begin{table}[t]
  \centering
  \small
  \setlength{\tabcolsep}{4pt}
  \caption{Occupation-exposure interaction estimates.}
  \label{tab:app_occ}
  \begin{tabular}{lrrrrr}
    \toprule
    Exposure specification & $\hat\theta$ & $p_{\mathrm{CRV1}}$ & $p_{\mathrm{WCB}}$ & $G^{*}$ & $N$ \\
    \midrule
    Published ILO index, weighted crosswalk       & $-0.618$ & 0.093 & 0.195 & 4.71 & 1{,}458 \\
    Published ILO index, unweighted crosswalk     & $-0.657$ & 0.080 & 0.163 & 4.66 & 1{,}458 \\
    Published ILO index, WLS by lagged employment & $-0.336$ & 0.376 & 0.440 & 3.05 & 1{,}458 \\
    AIOE-anchored rank proxy                      & $-0.892$ & 0.030 & 0.043 & 4.76 & 1{,}458 \\
    ILO-anchored rank proxy                       & $-0.717$ & 0.082 & 0.194 & 4.23 & 1{,}458 \\
    Cognitive versus manual                       & $-1.691$ & 0.026 & 0.063 & 5.81 & 1{,}296 \\
    \bottomrule
  \end{tabular}
\end{table}

The wild-cluster-bootstrap values in Table~\ref{tab:app_occ} use 99{,}999
Rademacher draws with the null imposed and seed 42; their Monte Carlo standard
errors range from 0.0006 to 0.0016.  The published ILO estimate is sensitive to
the 14 of 1{,}458 cells with $|\Delta\log E|>20$ log points: trimming those cells
moves the coefficient from $-0.62$ to $-0.33$ ($p_{\mathrm{WCB}}=0.34$).
Winsorising at $\pm20$ gives $-0.49$ ($p_{\mathrm{WCB}}=0.26$).  For the
author-coded AIOE rank, the corresponding trimmed and winsorised estimates are
$-0.42$ and $-0.66$, with bootstrap $p$-values 0.065 and 0.043.  Excluding the
COVID years, the largest economy, or 2021 preserves the sign of the published-index
interaction but gives CRV1 $p$-values 0.047, 0.080, and 0.103, versus bootstrap
$p$-values 0.191, 0.174, and 0.201.  Estonia contributes 37.9\% of the
unrestricted squared-score mass in the primary occupation specification.

A separate post-2022 design estimates $-0.405$ per SD for the interaction of
the published ILO index with the post indicator ($p=0.194$, $N=1{,}502$).
Interacting the published ILO exposure measure with national AI-hiring
intensity gives $-0.040$ ($p=0.824$, $N=1{,}286$).  These estimates do not
provide independent evidence of displacement.

\section{Auxiliary tests and inference diagnostics}
\label{app:aux}

\paragraph{Remaining preprocessing checks.}
Table~\ref{tab:app_preprocess} completes the preprocessing rows summarised in
Table~\ref{tab:robust}.  The sign is stable, while inference remains sensitive
to the reference distribution.  All wild-cluster-bootstrap values use 4{,}999
null-imposed Rademacher draws.

\begin{table}[t]
  \centering
  \small
  \setlength{\tabcolsep}{5pt}
  \caption{Additional preprocessing checks for the Panel A hiring coefficient.}
  \label{tab:app_preprocess}
  \begin{tabular}{lrrrrr}
    \toprule
    Specification & Estimate & 95\% CRV1 interval & $p_{\mathrm{CRV1}}$ & $p_{\mathrm{WCB}}$ & $N$ \\
    \midrule
    Winsorise employment growth & $-0.453$ & $[-0.927,\ 0.021]$ & 0.060 & 0.039 & 144 \\
    Winsorise hiring index      & $-0.562$ & $[-1.130,\ 0.006]$ & 0.052 & 0.048 & 144 \\
    Asinh migration             & $-0.481$ & $[-1.003,\ 0.042]$ & 0.070 & 0.050 & 144 \\
    Control for survey break    & $-0.476$ & $[-0.976,\ 0.024]$ & 0.061 & 0.037 & 144 \\
    \bottomrule
  \end{tabular}
\end{table}

\paragraph{Measures whose denominator is not total hiring.}
Table~\ref{tab:app_altmeasures} reports the denominator-independent probes.
Power is calculated against an effect of magnitude 0.474 per SD, using each
row's standard error and a two-sided $t(G-1)$ test.  The enterprise-AI-use
comparison is not informative: on its 96-observation common sample, relative
AI hiring is itself null ($p=0.655$), and the enterprise-use regression has
only 13\% power against the Panel A estimate.  The annual change in talent
concentration is the informative comparison: it has 90\% power and is close to
zero.  The three encodings are reported together to avoid choosing the only
well-powered one after seeing the estimates.

\begin{table}[t]
  \centering
  \small
  \setlength{\tabcolsep}{4pt}
  \caption{Denominator-independent measures and attainable power.}
  \label{tab:app_altmeasures}
  \begin{tabular}{lrrrrrr}
    \toprule
    Measure (per SD) & Estimate & SE & $p$ & $N$ & Power & MDE$_{80}$ \\
    \midrule
    Enterprise AI use, common sample & $-0.443$ & 0.545 & 0.425 & 96  & 0.13 & 1.596 \\
    Talent concentration, annual change & $+0.017$ & 0.141 & 0.908 & 221 & 0.90 & 0.411 \\
    Talent concentration, growth rate   & $-0.153$ & 0.184 & 0.414 & 221 & 0.70 & 0.537 \\
    Talent concentration, level         & $-0.347$ & 0.556 & 0.539 & 221 & 0.13 & 1.622 \\
    \bottomrule
  \end{tabular}
\end{table}

Adding Eurostat's annual job-vacancy rate does not discriminate between the
competing mechanisms.  On the Panel A vacancy subsample the hiring coefficient
moves from $-0.359$ ($p=0.081$, $N=131$, 22 countries) to $-0.366$ ($p=0.085$).
On the Panel B vacancy subsample it moves from $-0.0300$ ($p=0.057$, $N=120$,
24 countries) to $-0.0287$ ($p=0.054$).

\paragraph{Panel B and the 2021 break.}
Panel B retains the negative sign in the full 2021-2025 sample but does not
reproduce it after removing the survey-redesign year.  Coefficients in
Table~\ref{tab:app_panelb} are in raw index points.  The full-sample bootstrap
$p$-value is 0.0557 from 99{,}999 null-imposed Rademacher draws.

\begin{table}[t]
  \centering
  \small
  \setlength{\tabcolsep}{6pt}
  \caption{Panel B sensitivity to the 2021 Labour Force Survey break.}
  \label{tab:app_panelb}
  \begin{tabular}{lrrrrr}
    \toprule
    Sample & Estimate & Clustered SE & $p$ & $N$ & Countries \\
    \midrule
    Full, 2021-2025       & $-0.0274$ & 0.0138 & 0.0588 & 130 & 26 \\
    Excluding 2021         & $+0.0016$ & 0.0219 & 0.9427 & 104 & 26 \\
    Excluding break-flagged observations & $+0.0029$ & 0.0217 & 0.8961 & 102 & 26 \\
    \bottomrule
  \end{tabular}
\end{table}

\paragraph{Dynamics.}
Local projections estimate cumulative associations of $-0.471$, $-0.361$,
$-0.615$, and $-0.367$ at horizons zero through three.  The corresponding
CRV1 $p$-values are 0.063, 0.252, 0.041, and 0.161, while the sample falls from
144 to 72 observations.  The one-year lead is $-0.528$ with a 95\% interval
$[-1.386,0.330]$.  The non-monotone path, shrinking samples, and imprecise lead
do not support a dynamic causal interpretation and do not exclude reverse
causality.

\paragraph{Hierarchical model and flexible learners.}
A Bayesian hierarchical version of the baseline gives a posterior median of
$-0.43$ per SD for the hiring coefficient with 90\% interval $[-0.69,\,-0.19]$
($\widehat{R}\leq 1.005$, no divergent transitions), matching the least-squares
estimate. A random forest of 500 trees evaluated by leave-one-country-out
cross-validation attains RMSE 1.83 against 1.87 for the linear specification
and 1.95 for predicting the mean, so flexible learners recover no nonlinear
signal the linear model misses.

\paragraph{Influence concentration.}
Let $q_g$ be country $g$'s score contribution for the hiring coefficient and
let $s_g=q_g^2/\sum_h q_h^2$.  We report
\[
  G^{*}=\frac{1}{\sum_g s_g^2}
\]
as an inverse-Herfindahl diagnostic.  It should not be presented as a general degrees-of-freedom theorem.  In this fixed-effects design, the equicorrelated-cluster construction of Carter, Schnepel and Steigerwald \cite{carter2017}  is unavailable: the within-cluster vector of ones it is built on is annihilated by the country effects. The value reported here is therefore a squared-score plug-in.  For the baseline, unrestricted scores give $G^{*}=2.29$, with Estonia carrying
61.1\% of squared-score mass; null-restricted scores give 3.22, with T\"urkiye carrying 47.4\%.  Weighting countries by the within-cluster sum of squares of the residualised hiring regressor instead gives 5.52.  These alternatives show
that the exact scalar is estimator-dependent even though all diagnose severe concentration.  Table~\ref{tab:app_loco} reports the corresponding leave-one-country-out estimates.

\begin{table}[t]
  \centering
  \small
  \setlength{\tabcolsep}{5pt}
  \caption{Leave-one-country-out influence for the Panel A hiring coefficient.}
  \label{tab:app_loco}
  \begin{tabular}{lrrrr}
    \toprule
    Estimation sample & Estimate & CRV1 SE & $p$ & $N$ \\
    \midrule
    Full sample                  & $-0.474$ & 0.241 & 0.061 & 144 \\
    Drop T\"urkiye              & $-0.350$ & 0.191 & 0.080 & 138 \\
    Drop Latvia                  & $-0.497$ & 0.344 & 0.163 & 138 \\
    Drop Estonia                 & $-0.727$ & 0.233 & 0.005 & 138 \\
    Drop T\"urkiye and Latvia   & $-0.312$ & 0.254 & 0.232 & 132 \\
    \bottomrule
  \end{tabular}
\end{table}

The four reported reference distributions answer different questions and rest
on different assumptions.  The conventional CRV1 $t(23)$ reference gives
$p=0.0609$; the null-imposed wild-cluster bootstrap gives $p=0.0366$; and
studentized randomization inference over 20{,}000 permutations of whole-country
hiring trajectories gives $p=0.0450$.  The latter requires country trajectories
to be exchangeable conditional on the fixed-effect structure.  Referring the
same statistic to $t(G^{*}-1)$ gives $p=0.253$ using unrestricted
scores and $p=0.175$ using restricted scores, but those values are sensitivity
diagnostics rather than established finite-sample tests.  Using the
within-cluster-regressor weighting gives $p=0.112$.

\paragraph{Design-matched calibration.}
Table~\ref{tab:app_size} holds the observed design matrix fixed, imposes a zero
hiring effect, and draws errors independently within countries.  The first row
allows country-specific variances calibrated from the restricted residuals;
the second uses a common variance.  This is a favourable calibration, not a
model of arbitrary serial correlation.  It therefore cannot validate any
procedure for a general country panel.

\begin{table}[t]
  \centering
  \small
  \setlength{\tabcolsep}{5pt}
  \caption{Rejection rates at a nominal 5\% level (2{,}000 replications).}
  \label{tab:app_size}
  \begin{tabular}{lrrrr}
    \toprule
    Error process & CRV1 $t(23)$ & Wild bootstrap & $t(G^{*}-1)$ & Randomization \\
    \midrule
    Country-heteroskedastic & 0.051 & 0.054 & 0.014 & 0.041 \\
    Homoskedastic           & 0.067 & 0.046 & 0.032 & 0.041 \\
    \bottomrule
  \end{tabular}
\end{table}

Among the 353 heteroskedastic replications with score concentration at least as
severe as observed ($G^{*}\leq2.29$), rejection rates are 0.011
for CRV1, 0.031 for the wild bootstrap, and zero for
$t(G^{*}-1)$.  The plug-in itself is noisy: under this calibration
its median is 3.38 and its 5th-95th percentile range is 1.81-7.21, while the
observed 2.29 lies near the 17th percentile.  The simulations therefore show
that no single one of the four $p$-values should be treated as dispositive.

At the baseline CRV1 standard error of 0.2405, the closed-form 80\%-power
detection floor is 2.13 per SD at $G^{*}=2.29$, 1.19 at the
restricted-score value 3.22, 0.86 under the within-cluster-regressor weighting
5.52, and 0.70 if all 24 countries were equally effective.  Exact inversion of
the noncentral $t$ gives 2.38 and approximately 0.87 at the two ends of the
diagnostic range.  These calculations describe the current allocation of
identifying variation.  Additional observations concentrated in the same
influential countries need not improve the floor; adding countries or more
balanced within-country variation can.

\paragraph{Simulation behind Figure~\ref{fig:ident}.}
The numerical illustration fixes the true numerator-specific effect at zero,
sets $\gamma=0.6$, and draws 200{,}000 observations at each point of the grid.
It uses
\[
  R=(\gamma-1)\Delta\log H+u,\qquad
  g=0.8\Delta\log H-0.5s+\varepsilon,
\]
with independent standard-normal $\Delta\log H$ and $s$, Gaussian
$\varepsilon$ with standard deviation 0.3, and
$u=v+\rho s$ residualised on $\Delta\log H$.  The grid is
$\rho\in[-2.5,2.5]$.  The index-only coefficient ranges from $-0.35$ to
$+0.07$ even though the true effect is zero; after observing the denominator,
the coefficient on $u$ ranges from $-0.18$ to $+0.18$.  Because changing
$\rho$ changes both $\operatorname{Cov}(u,\eta)$ and
$\operatorname{Var}(u)$, the plotted curves are nonlinear.  The exercise is an
existence illustration, not a calibrated claim about the empirical frequency
or magnitude of either bias.

\paragraph{Event-study and shift-share diagnostics.}
The annual joint pre-trend test rejects for the AI-hiring level under
employment exposure ($p=0.010$), the change in AI hiring under employment
exposure ($p=0.015$), and the AI-hiring level under digital-core exposure
($p=0.024$).  It does not reject for total employment growth under employment
exposure ($p=0.100$), whose post-2022 coefficients are $+0.436$, $+0.014$, and
$+0.313$ and individually insignificant.  A separate monthly AI-hiring event
study rejects its joint pre-trend test ($p=0.017$).  There are four annual
designs and one separate monthly design.

The shift-share first stage is $+0.377$ ($p=0.025$, $R^2=0.023$) in a
bivariate regression and $-0.790$ ($p=0.031$) after adding the reported
controls.  The reduced-form coefficients are $-11.276$ for AI-hiring growth
($p=0.051$) and $+1.154$ for employment growth ($p=0.267$); the resulting IV
coefficient of AI hiring on employment growth is $-0.043$ ($p=0.905$).  The
first-stage sign reversal and rejected pre-trends preclude a causal reading.

\section{Specification curve for the primary coefficient}
\label{app:spec}

Figure~\ref{fig:spec} evaluates the defined specification set used in the
supplement: 192 Panel A combinations of controls, sample restrictions,
employment weights, winsorisation, and country-specific trends, together with
eight Panel B combinations of the 2021 restriction, vacancy control, and
employment weights.  Panel B coefficients are rescaled to percentage points
per sample SD of the hiring index.  All specifications include country and year
effects and use country-clustered CRV1 intervals with $t(G-1)$ critical values.

Of the 200 point estimates, 198 are negative.  Their median is $-0.533$ and
their range is $[-1.513,+0.097]$; the declared Panel A primary estimate
($-0.474$) is at the 62nd percentile.  These summaries describe sign and
magnitude sensitivity only.  The specifications are highly dependent, were
not a set of independently preregistered tests, and the display supplies no
familywise or selective-inference correction.  Accordingly, counts of rows
with $p<0.05$ are intentionally not used as evidence.

\begin{figure}[t]
  \centering
  \includegraphics[width=\linewidth]{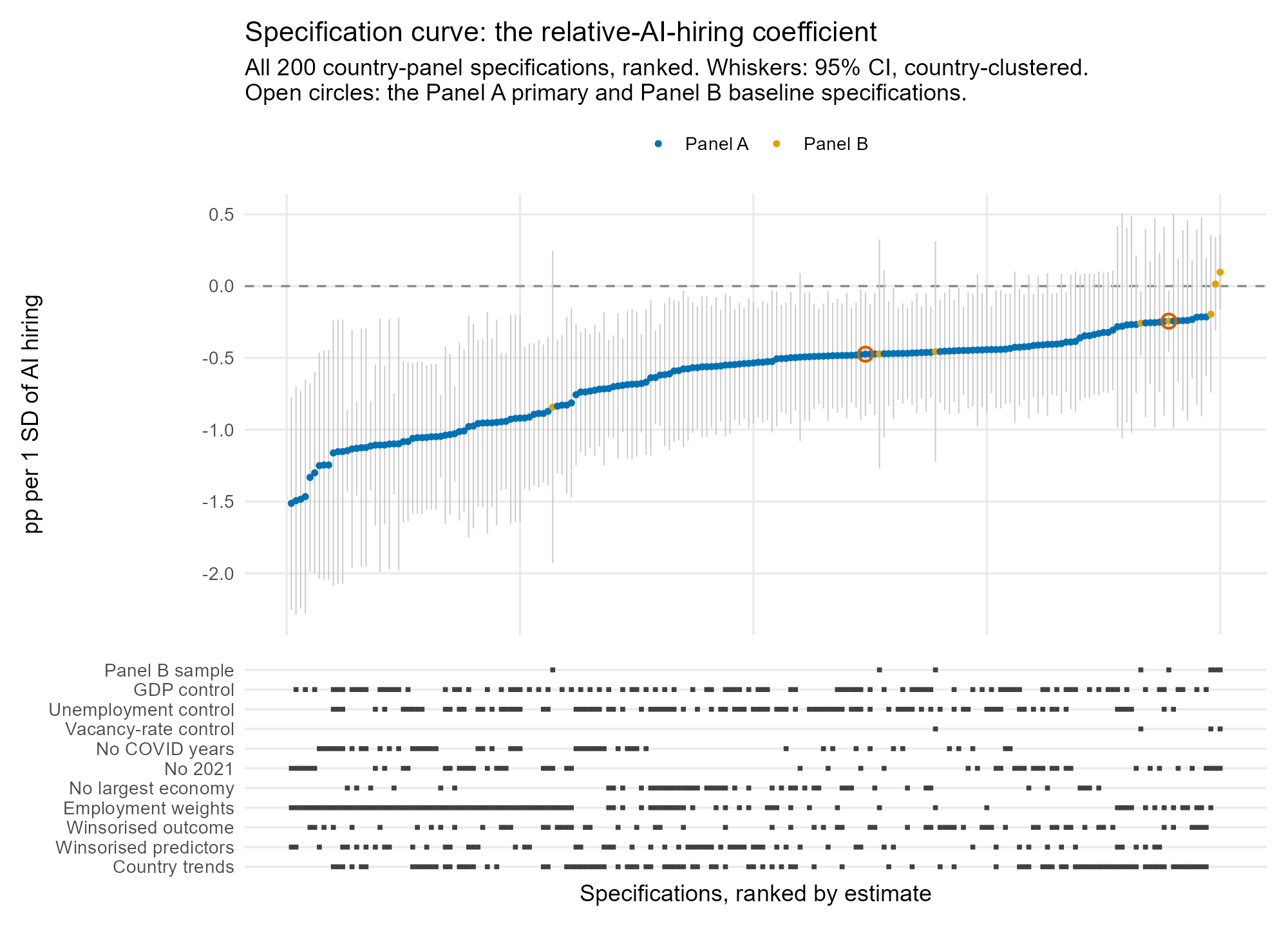}
  \caption{Relative-AI-hiring coefficients across 200 country-panel
  specifications.  Whiskers are 95\% country-clustered CRV1 intervals using
  $t(G-1)$ critical values.  Open circles mark the Panel A primary and Panel B
  baseline specifications.}
  \label{fig:spec}
\end{figure}

\section{Sources, samples, and additional figures}
\label{app:sources}

\paragraph{Sample construction and versioning.}
Table~\ref{tab:app_samples} records the samples used by the reported empirical
blocks.  Panel A is rebuilt from frozen raw inputs; Panel B is independently
joined from the later root-level extracts.  OECD.AI omits covered
country-years with zero recorded venture-capital deals, so the scripts fill
those cells with zero; a country absent from the extract entirely remains
missing.  T\"urkiye is the only such case in Panel A and has no venture-capital
value in any of its six years.

\begin{table}[t]
  \centering
  \small
  \setlength{\tabcolsep}{4pt}
  \caption{Analysis samples.}
  \label{tab:app_samples}
  \begin{tabular}{p{0.20\linewidth}p{0.15\linewidth}rrp{0.38\linewidth}}
    \toprule
    Block & Years & Countries & $N$ & Construction \\
    \midrule
    Panel A country-year & 2019-2024 & 24 & 144 & Frozen Panel A hiring and migration extracts; primary country panel. \\
    Panel B country-year & 2021-2025 & 26 & 130 & Refreshed hiring and migration extracts; includes 2025. \\
    Occupation interaction & 2019-2024 & 27 & 1{,}458 & Nine ISCO-08 major groups; country-year, country-occupation, and occupation-year effects. \\
    Post-2022 occupation design & 2017-2025 & 27 & 1{,}502 & Published ILO index; observations touching a Eurostat ``b'' break removed. \\
    Descriptive country/industry panels & 2018-2025 & varies & varies & OECD.AI activity and venture-capital series joined to Eurostat employment. \\
    \bottomrule
  \end{tabular}
\end{table}

The two annualized relative-AI-hiring files are not interchangeable.  On 208
overlapping country-years, the Panel A and refreshed extracts correlate 0.447;
200 Panel A cells are positive versus 119 refreshed cells.  The former has
mean 15.24 and standard deviation 13.22, while the latter has mean 1.67 and
standard deviation 7.72.  The exact files are
\path{data/panel_a/input/ai_hiring_index.csv} (truncated SHA-256
\texttt{65594a5d8ffff942}) and \path{data/ai_hiring_index.csv}
(\texttt{54df61bf6e9e3719}).  We found no public version history or methodological
notice explaining the revision.  The package's \texttt{data/MANIFEST.csv}
records file hashes, providers, and redistribution status; Panel A's frozen
inputs were exported on 13 August 2026 and are mapped to their source objects in
\path{data/panel_a/PROVENANCE.md}.

The Arena battle log is not redistributed because it declares no licence.
\path{scripts/fetch_arena_battles.R} downloads it, and the manifest records
the truncated SHA-256 \texttt{d4fb8013b0306624}.  All randomized procedures
set seed 42; the baseline country-panel and occupation bootstraps use 99{,}999
draws, and the randomization-inference calculation uses 20{,}000 permutations.

\paragraph{Additional figures.}
The intervals in Figure~\ref{fig:industry} and
Figure~\ref{fig:identification} use the estimation package's default
small-sample correction, as those figures belong to the descriptive pipeline.
The primary country-panel tables, Figure~\ref{fig:spec}, and
Figure~\ref{fig:replication} use the explicit CRV1 convention described in the
main text.  Stating the convention here avoids treating the package default and
the paper's CRV1 calculation as the same standard error.

\begin{figure}[t]
  \centering
  \includegraphics[width=\linewidth]{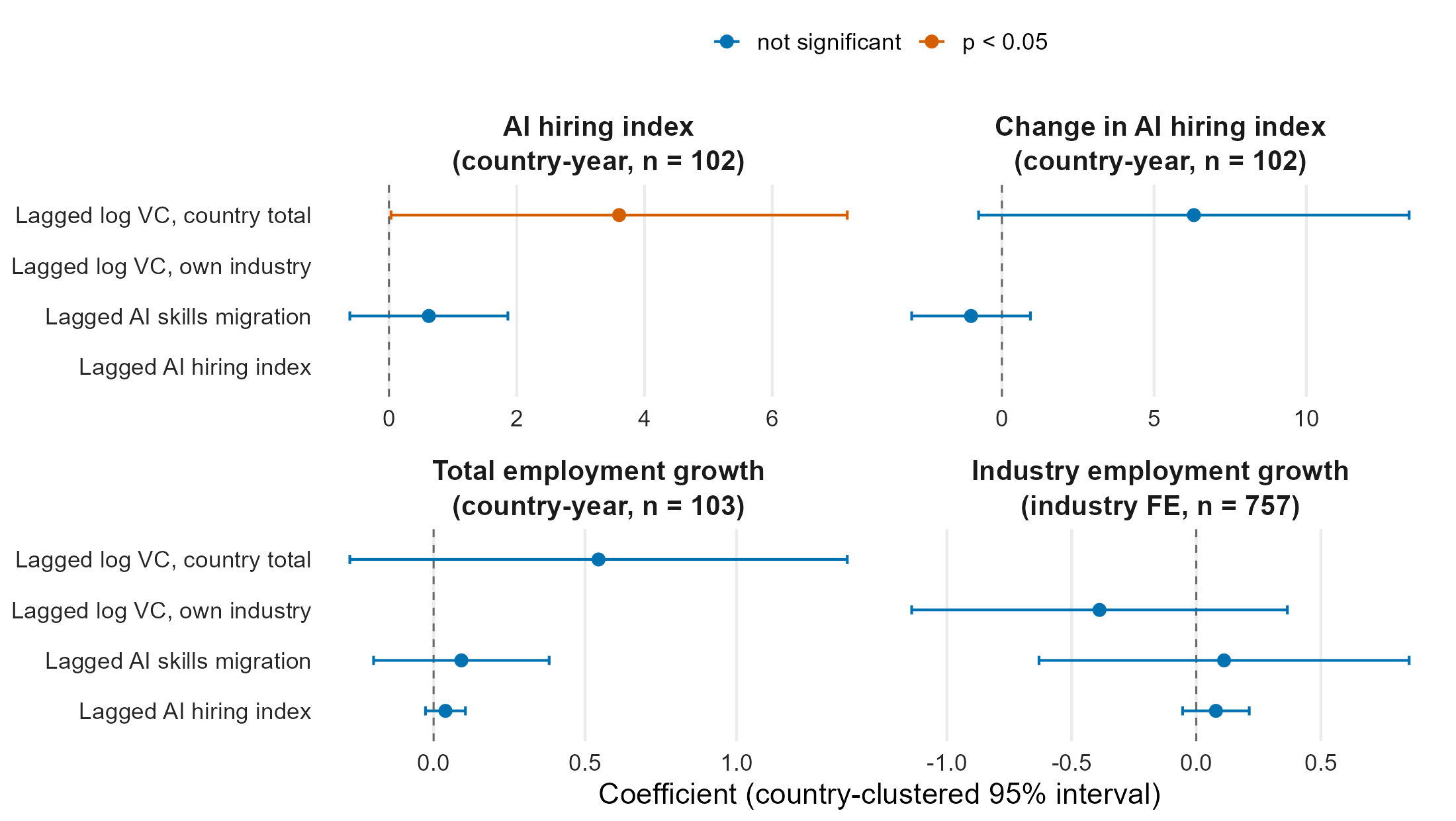}
  \caption{Four fixed-effects co-movement regressions for 2022-2025.  The
  first three are country-year models; the fourth is a country-industry-year
  model with industry effects.  Intervals are 95\% country-clustered intervals
  using the estimation package's default small-sample correction.  Only lagged
  country-total venture capital in the hiring-level regression excludes zero.}
  \label{fig:industry}
\end{figure}

\begin{figure}[t]
  \centering
  \includegraphics[width=0.66\linewidth]{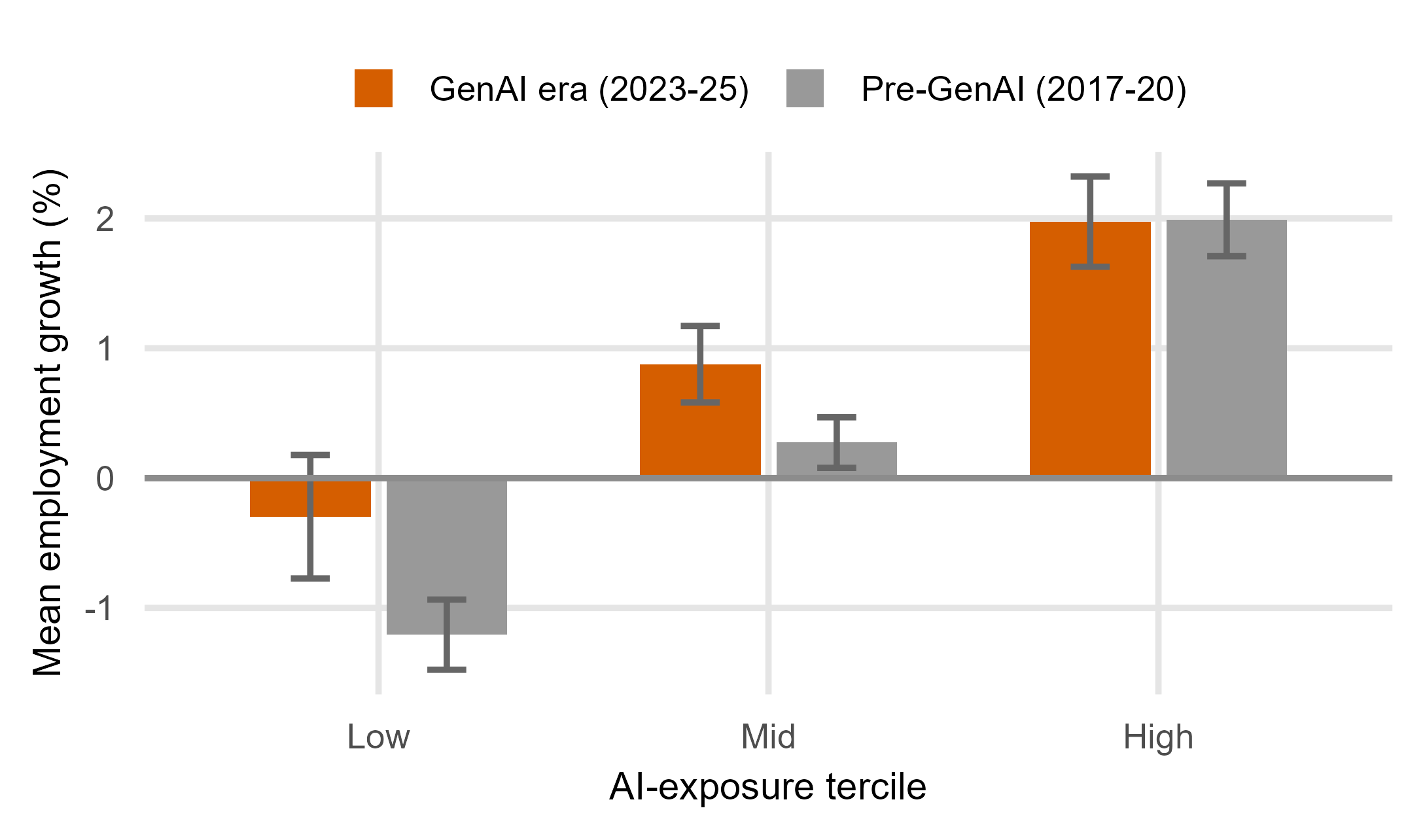}
  \caption{Mean occupation employment growth by tercile of the published ILO
  exposure measure before generative AI (2017-2020) and in 2023-2025.  Bars
  are one standard error across country means.  Growth is higher in both eras
  for the most exposed tercile, and the era-over-era increase is largest for
  the least exposed tercile.}
  \label{fig:terciles}
\end{figure}

\begin{figure}[t]
  \centering
  \includegraphics[width=0.70\linewidth]{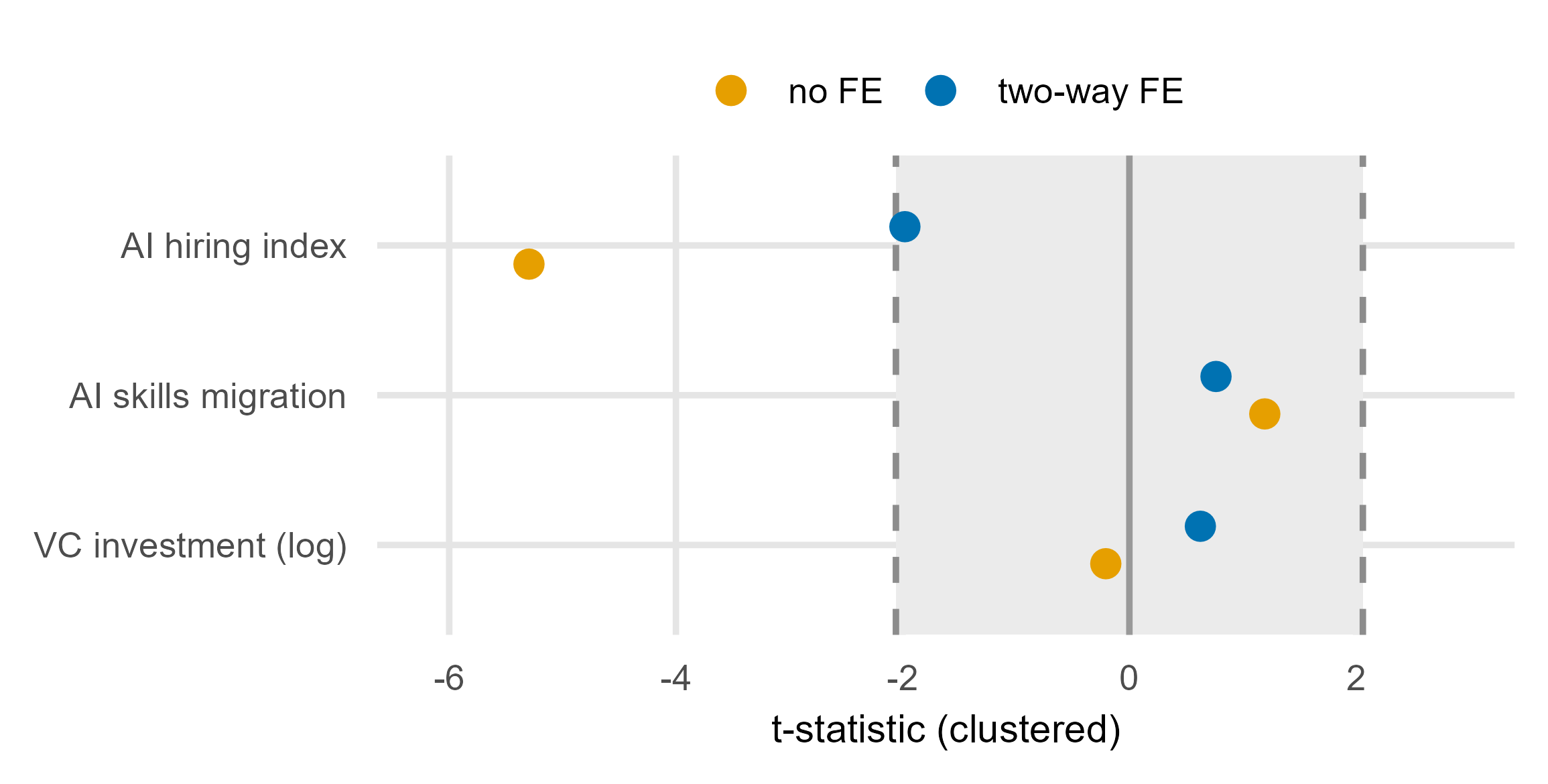}
  \caption{Panel B (26 countries, 2021-2025, $N=130$): clustered
  $t$-statistics for the three baseline predictors with and without country and
  year effects.  Dashed lines mark two-sided 5\% critical values based on
  $t(25)$.  The full-sample hiring coefficient is near, but does not cross, the
  threshold under two-way effects.}
  \label{fig:replication}
\end{figure}

\begin{figure}[t]
  \centering
  \includegraphics[width=\linewidth]{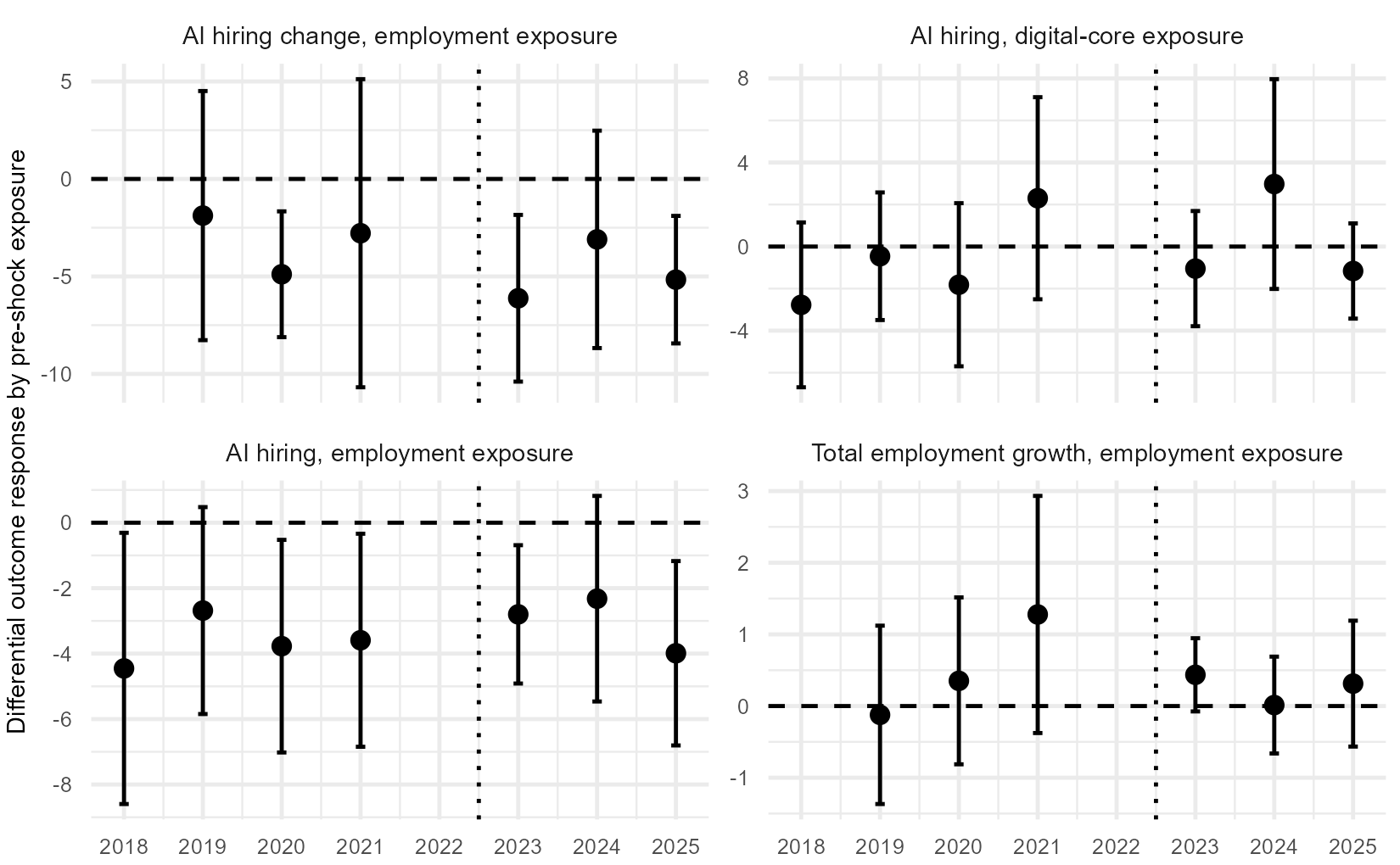}\par\medskip
  \includegraphics[width=\linewidth]{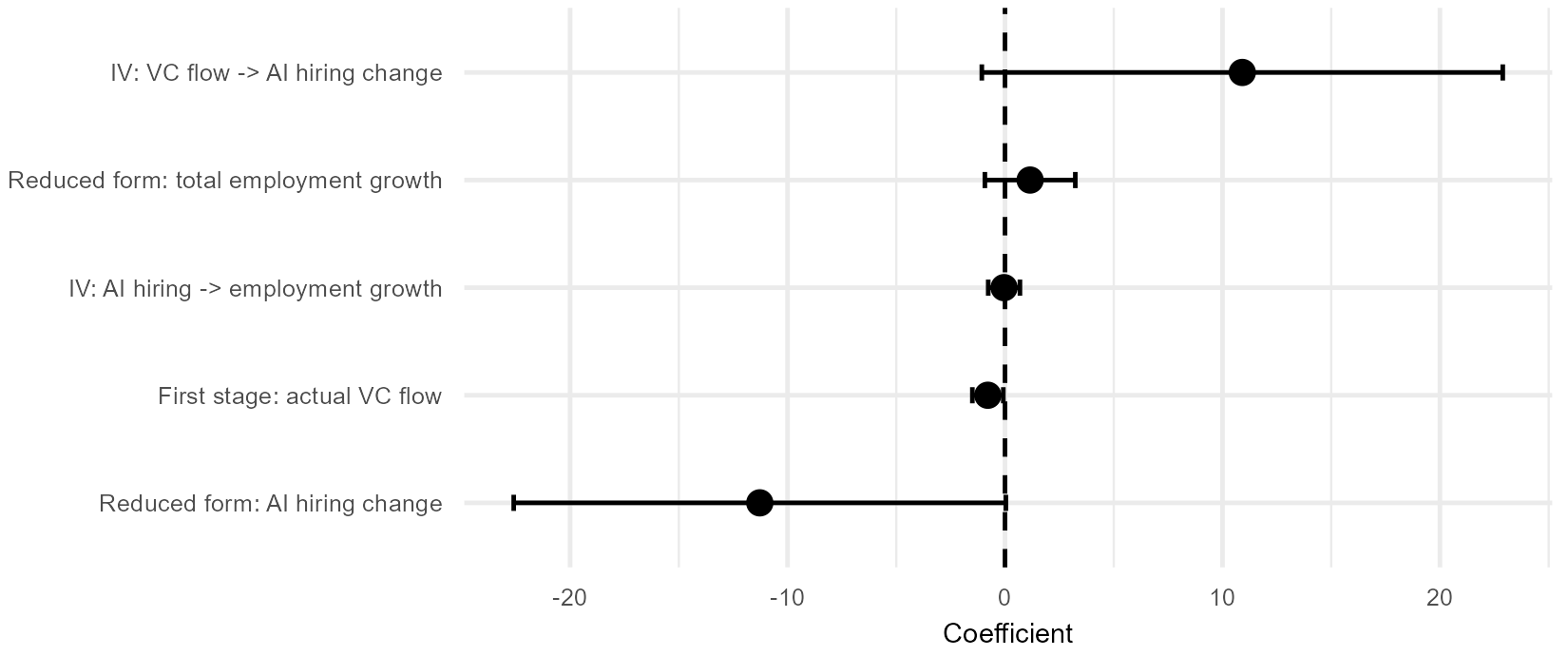}
  \caption{Identification diagnostics.  Top: annual event studies interacting
  year with 2018 employment exposure, with 2022 omitted.  Three of four annual
  specifications reject the joint pre-trend null.  Bottom: shift-share first
  stage, reduced forms, and IV estimates; the instrument uses 2018 main-industry
  employment shares interacted with leave-one-out European AI venture-capital
  main-industry shocks, and the first-stage coefficient changes sign when
  controls enter.  Intervals use the estimation package's default
  country-clustered small-sample correction.  These failures motivate a
  descriptive, not causal, interpretation.}
  \label{fig:identification}
\end{figure}

\end{document}